\documentclass[secnum,leqno]{rims-bessatsu}
\usepackage{amsmath,amssymb}

\newtheorem{thm}{Theorem}[section]

\newtheorem{lmm}[thm]{Lemma}

\theoremstyle{definition}
\newtheorem{dfn}[thm]{Definition}
\theoremstyle{remark}
\newtheorem{rem}[thm]{Remark}
\newenvironment{acknowledgement}{\par\medskip\emph{Acknowledgement.}}

\renewcommand\d{\mathrm d}
\newcommand\e{\mathrm{e}}

\def\le{\leqslant} \let\leq\le
\def\ge{\geqslant} \let\geq\ge

\def\Chi{\raisebox{.4ex}{$\chi$}}

\DeclareMathOperator{\tr}{tr\kern1pt}
\DeclareMathOperator{\supp}{\mathrm{supp}}
\DeclareMathOperator{\dist}{dist}
\DeclareMathOperator{\IE}{I\kern-.10em E}
\def\bigtimes{\mathop{\raisebox{-1.5pt}{\LARGE $\times$}}\limits}

\renewcommand{\AA}{\mathbb{A}}
\newcommand{\RR}{\mathbb{R}}
\newcommand{\Rd}{\mathbb{R}^d}

\newcommand{\Zd}{\mathbb{Z}^d}
\newcommand{\NN}{\mathbb{N}}
\newcommand{\CC}{\mathbb{C}}
\newcommand{\EE}{\mathbb{E}\mkern2mu}
\newcommand{\PP}{\mathbb{P}}

\newcommand{\cK}{\mathcal{K}}
\newcommand{\cU}{\mathcal{U}}
\newcommand{\cV}{\mathcal{V}}

\providecommand{\wtilde}{\widetilde}

\makeatletter

\newif\ifper\pertrue
\def\per{.}

\def\au#1#2{#2, #1}
\def\lau#1#2{#2, #1,}
\def\et{ and }
\def\ti#1{#1\ifper,\fi\pertrue}

\def\bti{\@ifnextchar[\bbti\bbbti}
\def\bbti[#1]#2{\emph{#2}, #1,}
\def\bbbti#1{\emph{#1,}}

\def\z{\@ifnextchar[\zz\zzz}
\def\zz[#1]#2#3#4#5{\perfalse\emph{#2}, \textbf{#3} (#5), #4 [#1]}
\def\zzz#1#2#3#4{\emph{#1}, \textbf{#2} (#4), #3\ifper\per\fi\pertrue}

\def\pub{\@ifstar\pubstar\pubnostar}
\def\pubnostar{\@ifnextchar[\@@pubnostar\@pubnostar}
\def\@@pubnostar[#1]#2#3#4{#2, #3, #4, #1\ifper\per\fi\pertrue}
\def\@pubnostar#1#2#3{#1, #2, #3\ifper\per\fi\pertrue}
\def\pubstar[#1]#2#3#4{\perfalse #2, #3, #4 [#1]\pertrue}

\makeatother

\title{Uniform convergence of spectral shift functions}
\author{Peter D.\ \textsc{Hislop}\footnote{Department of Mathematics,
    University of Kentucky, Lexington, Kentucky  40506-0027, USA.
    \newline e-mail: \texttt{hislop@ms.uky.edu}}
          ~and Peter \textsc{M\"uller}\footnote{Mathematisches Institut,
  Ludwig-Maximilians-Universit\"at, Theresienstra\ss{e} 39,
  80333 M\"unchen, Germany. \endgraf e-mail: \texttt{mueller@lmu.de}}}
\AuthorHead{Peter D.\ Hislop and Peter M\"uller}         
\support{PDH was partially supported by NSF grant 0803379.
PM acknowledges partial financial support by the German Research Foundation (DFG) through Sfb/Tr 12.}  
\VolumeNo{x}           
\YearNo{200x}           
\PagesNo{000--000}      
\communication{Version of 15 July 2010 
                                 }      
\begin{document}

\maketitle

\thispagestyle{empty}

\begin{abstract}
The spectral shift function $\xi_{L} (E)$ for a Schr\"odinger
operator restricted to a finite cube of length $L$ in multi-dimensional Euclidean space, with
Dirichlet boundary conditions, counts the number of eigenvalues
less than or equal to $E \in \RR$ created by a perturbation potential $V$.
We study the behavior of this function $\xi_L(E)$
as $L \rightarrow \infty$ for the case of a compactly-supported and bounded potential $V$.
After reviewing results of Kirsch [Proc.\ Amer.\ Math.\ Soc.\ \textbf{101}, 509--512 (1987)], and our recent pointwise convergence result for
the Ces\`aro mean [Proc.\ Amer.\ Math.\ Soc.\ \textbf{138}, 2141--2150 (2010)], we present a new result on the convergence of the energy-averaged spectral shift function that is uniform with respect to the location of the potential $V$ within the finite box.
\end{abstract}


\section{Statement of the Problem and Result}

In the late eighties, W.\ Kirsch \cite{Kir87, Kir89} investigated
the relative eigenvalue counting function $\xi_L (E)$ for compactly-supported
nonnegative perturbations $V$ of the nonnegative Laplacian $- \Delta_L \geq 0$ on cubes $\Lambda_L := ] -L / 2, L /2 [^d \subset \RR^d$ with Dirichlet boundary conditions.
The cubes $\Lambda_L$ of edge length $L>0$ are centered about the origin in $d$-dimensional Euclidean
space $\RR^{d}$. The real-valued local perturbation $V$ is
a compactly-supported, nonnegative potential $0 \le V \in
\mathrm{L}^{\infty}(\Rd)$. We emphasize that $V$ does not depend on
$L$ and, of course, $L$ is large enough so that $\supp(V) \subset
\Lambda_{L}$.

The finite-volume spectral shift function (SSF) or relative
eigenvalue counting function $\xi_L (E)$ for the pair of
self-adjoint operators $( - \Delta_L / 2 + V , - \Delta_L / 2)$ on the Hilbert space
$\mathrm{L}^2 (\Lambda_L)$, with Dirichlet boundary conditions on $\partial
\Lambda_L$, is defined as the real-valued function of $E\in\RR$ given by
 \begin{equation}
 \xi_L (E) := \# \bigg\{\begin{array}{l} \text{eigenvalues of $-\Delta_{L}/2$}\\
  \text{(counting\ multipl.)~ $\le E$}\end{array}\bigg\} - \# \bigg\{\begin{array}{l} \text{eigenvalues of $-\Delta_{L}/2 + V$}\\
  \text{(counting multipl.)~ $\le E$}\end{array}\bigg\} .
 \end{equation}
Kirsch was interested in the limit of $\xi_{L} (E)$ as $L \to\infty$. Perhaps surprisingly, Kirsch
proved \cite{Kir87} that if $d \in\NN\setminus\{1\}$ and if $\int_{\RR^{d}}\d x\,
V(x) >0$, then
  \begin{equation}\label{unbdd1}
    \sup_{L > 0} \xi_L (E) = \infty  \qquad\quad \text{for every~} E>0.
  \end{equation}
  Furthermore, he proved the existence of a countable dense set of
  energies $\mathcal{E} \subset [0, \infty[$ so that
  \begin{equation}\label{unbdd2}
    \sup_{L \in \NN} \xi_L (E) = \infty  \qquad\quad \text{for every~} E \in
  \mathcal{E}.
  \end{equation}
Clearly, the reason for this divergence is the growing degeneracy of
eigenvalues of the Laplacian in $d\ge 2$ dimensions as $L\to\infty$, which is
lifted by the perturbation $V$. In contrast, in $d=1$ space dimension, and also for
corresponding lattice systems in arbitrary dimension, the spectral shift function remains bounded in this limit,
as follows from a finite-rank-perturbation argument.

More generally, if one replaces the sequence $L_{n} =n \in\NN$ in \eqref{unbdd2} by another diverging sequence of lengths,
one would expect the set of ``bad'' energies $\mathcal{E}$ to change.
One might conjecture, however, that the largest set of energies $\mathcal{E}$, on which $\xi_{L}$ explodes, still has zero Lebesgue measure.
Although still unproven in this generality, this conjecture is strongly supported by Theorems~\ref{HiMu-1} and~\ref{HiMu-2} below.

In order to state these theorems, we need to list the hypotheses.
We write $\mathcal{K}(\RR^{d})$ and $\mathcal{K}_{\mathrm{loc}} (\RR^d)$ to denote the Kato class and the local
Kato class, respectively \cite{AiSi82, Sim82}. We say $U$ is Kato decomposable, $U\in\cK_{\pm}(\Rd)$, if
$\max\{0, U\} \in \cK_{\mathrm{loc}}(\RR^{d})$ and $\max\{0, -U\} \in \cK(\RR^{d})$.
In the following we consider two real-valued potential functions $U$ and $V$ on $\RR^{d}$ such that
\begin{equation}
  \label{ass}
  \tag{$\star$}
  U \in \cK_{\pm}(\RR^{d}), \qquad  0\le V \in\cK_{\mathrm{loc}}(\RR^{d}) , \quad \supp (V) \subseteq \Lambda_{\ell} \text{~ for some~} \ell>0.
\end{equation}
 We also introduce the
corresponding infinite-volume self-adjoint Schr\"odinger operators $H_{0} :=
-(\Delta/2) + U$ and $H_{1}:= H_{0} + V$ on $\mathrm{L}^{2}(\RR^{d})$.
Their self-adjoint finite-volume Dirichlet restrictions $H_{0}^{(L)}$ and
$H_{1}^{(L)}$ to $\mathrm{L}^{2}(\Lambda_{L})$ have compact resolvents and, therefore, discrete spectrum. For a given energy $E \in \RR$,
let $N_0^{(L)}(E)$, resp.\ $N_1^{(L)}(E)$, denote the number of eigenvalues,
including multiplicity, for $H_0^{(L)}$, resp.\ $H_1^{(L)}$, less than or
equal to $E$.  These are both monotone increasing functions of the energy $E$.
We define the relative eigenvalue counting function by
\begin{equation}\label{ssf2}
  E \mapsto \xi_L (E) \equiv \xi (E; H_1^{(L)}, H_0^{(L)})
  := N_0^{(L)}(E) -N_1^{(L)}(E) \geq 0
\end{equation}
for all $E \in \RR$.  It is known that this function is equal to the (more
generally defined) \emph{spectral shift function} for the pair $(H_1^{(L)},
H_0^{(L)})$, see e.g.\ \cite{Yaf92, BiYa93} or Eq.\ (5.1) in the Appendix of \cite{HiMu10}.

\begin{thm}
 \label{HiMu-1}
  Let $d\in\NN$ and assume \eqref{ass}. Then, we have
  \begin{equation}\label{ssf-bound2}
    \lim_{L \rightarrow \infty} \int_\RR \!\d E\; \xi_{L}(E) \, f(E) =
    \int_\RR\!\d E \; \xi(E)\; f(E)
  \end{equation}
 for every function $f$ of the form $f=\Chi_{I} g$, where $g \in C(\RR)$ is continuous and $\Chi_{I}$ is the indicator function of a (finite) interval $I \subset\RR$. In particular, for Lebesgue-almost all $E \in \RR$ we have
  \begin{equation}\label{eq-ssf4}
  \lim_{\delta \downarrow 0} \lim_{L \rightarrow \infty}
    \frac{1}{\delta} \int_{E}^{E+\delta} \!\d E'\, \xi_{L}(E') = \xi(E).
  \end{equation}
\end{thm}

We refer to \cite{HiMu10} for a proof of the theorem, see also \cite{GeKoSc95}.
Kirsch's result \eqref{unbdd1} shows that one cannot get rid of the
energy smoothing in \eqref{eq-ssf4}, that is, the limits $\delta \downarrow
0$ and $L\to\infty$ must not be interchanged. The best one could hope for is
convergence Lebesgue-almost everywhere  of $(\xi_{L_{j}})_{j\in\NN}$ for
sequences of diverging lengths. The next theorem is a
partial result in this direction.
\begin{thm}
  \label{HiMu-2}
  Let $d\in\NN$ and assume \eqref{ass}. Then, for every sequence of lengths
  $(L_{j})_{j\in\NN}$ $\subset \, ]0,\infty[$ with $\lim_{j\to\infty} L_{j} =
  \infty$ there exists a subsequence $(j_{i})_{i\in\NN} \subset\NN$ with
  $\lim_{i\to\infty} j_{i} =\infty$ such that for every subsequence
  $(i_{k})_{k\in\NN} \subset\NN$ with $\lim_{k\to\infty} i_{k} =\infty$ we
  have
  \begin{equation}
    \label{ces-in}
    \lim_{K\to\infty} \frac{1}{K} \sum_{k=1}^{K} \xi_{\wtilde{L}_{k}}(E) \leq \xi(E)
  \end{equation}
  for Lebesgue-almost all $E\in\RR$. Here we have set $\wtilde{L}_{k}
  :=L_{j_{i_{k}}}$ for all $k\in\NN$.
\end{thm}

The simple proof \cite{HiMu10} of Theorem~\ref{HiMu-2} relies on Theorem~\ref{HiMu-1}, the fact that $V$, and hence $\xi_{L}$, has a definite sign, and a deep result of Koml\'os. For one-dimensional systems, equality in (\ref{ces-in}) has recently been shown in \cite{BoMa09}. We refer to the literature cited in \cite{HiMu10} for further estimates of spectral shift functions.

In this note we are interested in refining the convergence in Theorem~\ref{HiMu-1} for the case where the background potential $U$ is periodic with period $p = (p_{1}, \ldots, p_{d}) \in\Rd$. In this situation, one can consider the potential
$V$ centered at any point $x_0 + x \in \Lambda_{L}$, with $x \in p\Zd$ and $x_0 \in \Lambda_L$, so that $
\Lambda_\ell (x_0) \subset \Lambda_L$. 
One expects that the limit in \eqref{ssf-bound2} is independent of the shift $x \in p \Zd$. We prove that this is true and that the 
convergence in \eqref{ssf-bound2} is uniform in the shift vectors $x \in p \Zd$.

To formulate this precisely, we introduce some more notation.
Recall that $\ell >0$ was defined by $\supp(V) \subseteq \Lambda_{\ell}$.
We choose a security distance function $D$ with the properties  $D(L) \le (L -\ell)/2$ for
all $L\ge \ell$ and $\lim_{L\to\infty} D(L) =\infty$.
For $x_{0} \in\RR^{d}$ we define the set of allowed shifts
\begin{equation}
\label{shift}
 A_{L}(x_{0}) := \Big\{ x_{0}+ x : x \in p\Zd \text{~and~} \dist\big(\Lambda_{\ell}(x_{0} + x), \RR^{d} \setminus \Lambda_{L}\big) > D(L) \Big\}
\end{equation}
in the box $\Lambda_{L}$. Here we used the notation $\Lambda_{\ell}(y)$ to indicate that the cube is centered about $y\in\Rd$ instead of the origin. Shifts in $A_{L}(x_{0})$ differ only by a multiple of the period of $U$, and it is enough to consider $x_{0} \in \raisebox{-1pt}{\Large$\times$}_{j=1}^{d} [0, p_{j}[$ in the periodicity cell of $U$.
Given $y\in\RR^{d}$ we write $V_{y} := V(\,\pmb{\cdot}\, -  y)$ for the shifted perturbation potential and note that if $y\in A_{L}(x_{0})$, then $\supp(V_{y})$ respects the security distance $D(L)$ to the boundary of $\Lambda_{L}$. Finally we define the Schr\"odinger operator $H_{1,y} := H_{0} + V_{y}$ with the shifted perturbation potential along with its finite-volume Dirichlet restrictions $H_{1,y}^{(L)}$. We do not require $V\ge 0$ any more.

\begin{thm}
  \label{uniform}
    Let $d\in\NN$ and let $U,V \in \cK_{\pm}(\Rd)$. Assume that $U$ is periodic with period $p\in\Rd$ and $\supp(V) \subseteq \Lambda_{\ell}$ for some $\ell >0$. Fix an arbitrary point $x_{0} \in \raisebox{-1pt}{\Large$\times$}_{j=1}^{d} [0, p_{j}[$ in the periodicity cell of $U$ and a sequence $(L_{n})_{n\in\NN} \subset[\ell, \infty[$ of diverging lengths, $\lim_{n\to\infty} L_{n} = \infty$.
Then,
  \begin{equation}\label{ssf-uniform}
    \lim_{n \rightarrow \infty} \int_\RR \!\d E\; \xi\big(E; H_{1, x_{L_{n}}}^{(L_{n})}, H_{0}^{(L_{n})}\big) \, f(E) =
    \int_\RR\!\d E \; \xi(E; H_{1, x_{0}}, H_{0})\; f(E)
  \end{equation}
holds uniformly in $(x_{L_{n}})_{n\in\NN} \in \raisebox{-1pt}{\Large$\times$}_{n\in\NN}  A_{L_{n}}(x_{0})$ for every given function $f$ of the form $f=\Chi_{I} g$, where $g \in C(\RR)$ is continuous and $\Chi_{I}$ is the indicator function of a (finite) interval $I \subset\RR$.
\end{thm}

The proof of Theorem~\ref{uniform} relies on a suitable continuity theorem for Laplace transforms, Theorem \ref{cont-laplace}, proven in the next section. We prove Theorem \ref{uniform}  in Section \ref{sec:assumptions1} by verifying the assumptions
of Theorem \ref{cont-laplace} using the Feynman-Kac representation for Schr\"o\-ding\-er semigroups.

\begin{rem}
 Theorem~\ref{uniform} also works in the context of ergodic random Schr\"o\-dinger operators, if we replace the periodic potential $U$ by a random potential $U_{\omega}$ that is ergodic with respect to a subgroup $G$ of the translation group $\Rd$. In this case the group $G$ replaces $p\Zd$ in the definition \eqref{shift} of the sets $A_{L}(x_{0})$ of allowed shift vectors. In addition, we have to replace the spectral shift functions in \eqref{ssf-uniform} by their expectations $\IE$ over randomness. In fact, our original motivation for Theorem~\ref{uniform} stems from the construction of a strictly positive lower bound on the density of states for alloy-type random Schr\"odinger operators  in the continuum \cite{HiKlMu10}, where it will be needed. As is apparent already in the lattice case of this problem \cite{HiMu08}, one needs to control a Ces\`aro sum
\begin{equation}
\label{cesaro}
 \frac{1}{(L/\ell)^{d}} \sum_{j \in \ell \Zd \cap \Lambda_{L}} \int_{\RR}\!\d E\, f(E) \, \IE \big[\xi\big(E; H_{1,j}^{(L)}, H_{0}^{(L)}\big) \big]
\end{equation}
 in the limit $L\to\infty$ for some given function $f$ as allowed by Theorem~\ref{uniform}. We apply Theorem~\ref{uniform} with the distance function $D(L) = (1/2)\log(L -\ell +1)$ to \eqref{cesaro}. Then the asserted uniformity of the convergence in $j$ implies
\begin{equation}
    \lim_{L\to\infty}\frac{\ell^{d}}{L^{d}} \sum_{j \in \ell \Zd \cap \Lambda_{L}} \int_{\RR}\!\d E\, f(E) \, \IE \big[\xi  \big(E; H_{1,j}^{(L)}, H_{0}^{(L)}\big) \big]
    = \int_{\RR}\d E\, f(E)\, \IE \big[\xi\big(E; H_{1,0}, H_{0}\big) \big].
\end{equation}
\end{rem}

%
\section{An abstract uniform convergence result for measures}
%

The uniform convergence of the finite-volume spectral shift
functions will follow from a more general result, Theorem \ref{cont-laplace}, on the continuity of the Laplace transform of Borel measures.
We begin with a simple observation:

\begin{lmm} \label{wonderland}
 For given sets $A_{n}$, $n\in \NN$, define $\AA := \raisebox{-1pt}{\Large$\times$}_{n \in\NN} A_{n}$. Consider a family of sequences in $\CC$ which is indexed by $a \equiv (a_{n})_{n\in\NN}\in \AA$ and has the property that the $n$-th sequence element depends only on $a_{n}$, but not on $a_{m}$ for $m\neq n$. We denote a sequence of this family by $(x_{n}^{a_{n}})_{n\in\NN}$. Suppose that for every $a\in\AA$  the limit
\begin{equation}
 x := \lim_{n\to\infty} x_{n}^{a_{n}}
\end{equation}
exists in $\CC$ and is independent of $a\in \AA$. Then the sequence $(x_{n}^{a_{n}})_{n\in\NN}$ converges to $x$ uniformly in $a\in\AA$.
\end{lmm}

\begin{proof}
 We argue by contradiction. Assume that the limit $x$ is not approached uniformly in $a$. Then there is $\varepsilon >0$ such that for every $N\in\NN$ there is $n \ge N$ and $\alpha_{n} \in A_{n}$ such that $|x_{n}^{\alpha_{n}} - x| > \varepsilon$. In other words, there exists $\varepsilon >0$, a subsequence $(n_{j})_{j\in\NN}$ with $\lim_{j\to\infty}n_{j} = \infty$ and $\alpha_{n_{j}} \in A_{n_{j}}$ for all $j\in\NN$ such that
\begin{equation}
\label{contra}
 |x_{n_{j}}^{\alpha_{n_{j}}} - x| > \varepsilon \qquad\quad\text{for all~} j\in\NN.
\end{equation}
Define $a\in\AA$ by setting $a_{n_{j}} := \alpha_{n_{j}}$ for all $j\in\NN$ and $a_{k}$ arbirtrary for $k \notin \{n_{j}: j\in\NN \}$. But then $\lim_{n\to\infty} x_{n}^{a_{n}} = x$ by hypothesis, and hence $\lim_{j\to\infty} x_{n_{j}}^{a_{n_{j}}} = x$, contradicting \eqref{contra}.
\end{proof}

\begin{dfn}
 Let $\mu$ be a Borel measure on $\RR$. If there is $t_{0} \in\RR$ such that the integral
\begin{equation}
    \wtilde{\mu}(t) := \int_{\RR}\!\d\mu(x)\, \e^{-t x} < \infty
\end{equation}
 is finite for all $t \ge t_{0}$, we say that the \emph{(two-sided) Laplace transform} of $\mu$ exists for $t\ge t_{0}$.
\end{dfn}

We will need the following version of a continuity theorem for Laplace transforms.

\begin{thm} \label{cont-laplace}
Let $(\mu_{n}^{a_{n}})_{n\in\NN}$ be a sequence of Borel measures on $\RR$ for every $a \equiv(a_{n})_{n\in\NN} \in \AA$. Assume that for every $a\in\AA$ there exists $t_{a} \in\RR$ such that for all $n\in\NN$ the Laplace transform of $\mu_{n}^{a_{n}}$ exists for $t \ge t_{a}$. Suppose further that for every $a\in\AA$ and every $t \ge t_{a}$ the limit
\begin{equation}
 \wtilde{\mu}(t) := \lim_{n\to\infty} \wtilde{\mu}_{n}^{a_{n}}(t)
\end{equation}
 exists in $\RR$ and is independent of $a\in \AA$. Then $\wtilde{\mu}$ is the Laplace transform
of a Borel measure $\mu$ on $\RR$ and $\mu_{n}^{a_{n}}$ converges vaguely to $\mu$ as $n\to\infty$, the convergence being uniform in $a\in\AA$. In other words,
\begin{equation} \label{cont-claim}
 \lim_{n\to\infty} \int_{\RR}\!\d\mu_{n}^{a_{n}}(x)\, f(x) =
  \int_{\RR}\!\d\mu(x)\, f(x)
\end{equation}
holds uniformly in $a\in\AA$ for every given function $f$ of the form $f=\Chi_{I} g$, where $g \in C(\RR)$ and $\Chi_{I}$ is the indicator function of a (finite) interval $I \subset\RR$ whose endpoints are not charged by the measure $\mu$.
\end{thm}

\begin{proof}
Existence of the limiting measure $\mu$ and pointwise convergence for every $a\in\AA$ of the limit in \eqref{cont-claim} for every given $f$ of the specified form is a standard continuity theorem for Laplace transforms, see for example \cite[Thm.\ 2a in Sect.\ XIII.1]{Fel71} and replace one-sided by two-sided Laplace transforms there. Uniformity of the convergence in $a\in \AA$ then follows from Lemma~\ref{wonderland} applied to the sequence $\big(\int_{\RR}\!\d\mu_{n}^{a_{n}}(x)\, f(x)\big)_{n\in\NN}$.
\end{proof}

%
\section{Proof of Theorem~\ref{uniform}}\label{sec:assumptions1}
%

We prove Theorem \ref{uniform} on the uniform convergence of the finite-volume spectral shift
functions using Theorem \ref{cont-laplace}.

\begin{proof}[Proof of Theorem~\ref{uniform}] 
1. We verify the assumptions of Theorem~\ref{cont-laplace}.
To this end we fix $x_{0}$ in the periodicity cell of $U$ and a
sequence $(L_{n})_{n\in\NN}$ of diverging lengths. Thanks to Dirichlet-Neumann bracketing and since both $U$ and $V$ are Kato decomposable, we have
$\mathrm{inf\,spec\,} H_{1, x_{L_{n}}}^{(L_{n})} \ge \mathrm{inf\,spec\,} H_{1, x_{0}} > -\infty$ for all $x_{L_{n}} \in  A_{L_{n}}(x_{0})$ and all $n\in\NN$. From this we infer that for every $n\in\NN$ and for
every shift $x_{L_{n}} \in  A_{L_{n}}(x_{0})$ the two-sided Laplace
transform
\begin{equation}
\begin{split}
 \wtilde{\xi}^{(n)}_{x_{L_{n}}}(t) &:= \int_{\RR}\!\d E\, \e^{-tE} \,
 \xi\big(E; H_{1,x_{L_{n}}}^{(L_{n})}, H_{0}^{(L_{n})}\big) \\
 &\phantom{:}=  \frac{1}{t} \;\tr_{\mathrm{L}^{2}(\Lambda_{L_{n}})}
\left[ \exp\big(-t H_{0}^{(L_{n})} \big) - \exp\big(-t H_{1, x_{L_{n}}}^{(L_{n})} \big) \right]
\end{split}
 \end{equation}
exists for every $t >0$. So fix $t>0$ from now on. 

\par\noindent
2. The standard Feynman--Kac representation \cite{Sim79} of the heat kernel gives
\begin{equation}
\begin{split}
   \varphi^{(n)}_{x_{L_{n}}}(t) &:=  t (2 \pi t)^{d/2} \; \wtilde{\xi}^{(n)}_{x_{L_{n}}}(t)\\
  &\phantom{:}= \int_{\Lambda_{L_{n}}} \!\d x \;
  ~\EE_{x,x}^{0,t} \left[ \Chi_{\Lambda_{L_{n}}}^{\,t} (b) \, \e^{- \int_0^t \d s \,
      U(b(s) )} \left( 1 - \e^{- \int_0^t \d s \, V_{x_{L_{n}}}(b(s)) } \right) \right]  \\
   &\phantom{:}= \int_{\Lambda_{L_{n}}} \!\d x \;
  ~\EE_{x,x}^{0,t} \left[ \Chi_{\Lambda_{L_{n}}}^{\,t} (b) \,  \cU_{t}(b) \,\cV_{t}(b- x_{L_{n}}) \right].
\end{split}
 \end{equation}
Here, $\EE_{x,y}^{0,t}$ denotes the normalized expectation over all Brownian
bridge paths $b$ starting at $x\in\Rd$ at time $s=0$ and ending at
$y\in\Rd$ at time $s=t$. The Dirichlet boundary condition is taken into
account by the cut-off functional $\Chi_{\Lambda}^{\,t}(b)$, which is equal to
one if $b(s) \in \Lambda$ for all $s \in [0, t]$, and zero otherwise.
Moreover, we have introduced the Brownian functionals
$\cU_{t}(b) := \e^{- \int_0^t \d s \, U(b(s))}$ and $\cV_{t}(b) := 1- \e^{- \int_0^t \d s \, V(b(s))}$.

\par\noindent
3. We shift the integration variables $x \rightarrow x + x_{L_{n}}$ and $b \rightarrow b + x_{L_{n}}$. This and the periodicity of $U$ results in
\begin{equation}
\label{conv-start}
\begin{split}
   \varphi^{(n)}_{x_{L_{n}}}(t)  &= \int_{\Rd} \!\d x \, \Chi_{\Lambda^{n}}(x) \; \EE_{x,x}^{0,t} \left[
   \Chi_{\Lambda^{n}}^{\,t} (b) \, \cU_{t}(b+x_{0}) \, \cV_{t}(b) \right]
   = \int_{\Rd} \!\d x \, \Chi_{\Lambda^{n}}(x) \, F_{n}(x)
\end{split}
\end{equation}
where we have introduced the abbreviations $\Lambda^{n} := \Lambda_{L_{n}}(-x_{L_{n}})$ and
\begin{equation}\label{eq:functional1}
 F_{n}(x) := \EE_{x,x}^{0,t} \left[
   \Chi_{\Lambda^{n}}^{\,t} (b) \, \cU_{t}(b+x_{0}) \, \cV_{t}(b) \right]  .
\end{equation}

\par\noindent
4. We evaluate the limit as $n \rightarrow \infty$ of $F_n (x)$  in (\ref{eq:functional1}).
We observe that $\Lambda^{n} \supseteq \Lambda_{\ell +2 D(L_{n})}$,
independently of $x_{L_{n}} \in A_{L_{n}}(x_{0})$ for every $n\in\NN$, whence
\begin{equation}
 \lim_{n\to\infty} \Chi_{\Lambda^{n}}(x) = 1 \qquad \text{and} \qquad
 \lim_{n\to\infty} \Chi^{\,t}_{\Lambda^{n}}(b) = 1
\end{equation}
for every $x\in\Rd$, $\PP_{0,0}^{0,t}$-a.e.\ Brownian bridge path $b$ and every choice of shifts
\begin{equation}
 (x_{L_{n}})_{n\in\NN} \in \bigtimes_{n \in\NN} A_{L_{n}}(x_{0}) =: \AA(x_{0}).
\end{equation}
Here, we have exploited the $\PP_{0,0}^{0,t}$-a.s continuity of $s \mapsto b(s)$.
Next, the expectation of the functiona $\mathcal{U}_t (b)$ is controlled
using \cite[Eqs.\ (6.20), (6.21)]{BrHuLe00} that state that
\begin{equation}
    \label{BHL}
 \sup_{x\in\Rd} \EE_{x,x}^{0,t} \big[ \e^{- \int_0^t \d s \, W(b(s))}\big] < \infty
\end{equation}
for every $W \in\cK_{\pm}(\Rd)$. This bound (\ref{BHL}) and dominated convergence yield the existence of
\begin{equation}
\label{inner-limit}
 F(x) := \lim_{n\to\infty} F_{n}(x) =  \EE_{x,x}^{0,t} \big[ \cU_{t}(b+x_{0}) \, \cV_{t}(b) \big]
\end{equation}
for every $x\in\Rd$ and every $(x_{L_{n}})_{n\in\NN} \in \AA(x_{0})$.

\par\noindent
5. Finally, for every $x\in\Rd$, every $n\in\NN$ and every $x_{L_{n}} \in A_{L_{n}}(x_{0})$ we have the estimate
\begin{equation}\label{eq:upperbd1}
    |F_{n}(x)| \le  G(x) := \EE_{x,x}^{0,t} \left[  \Chi_{\Xi_{x}^{t}}(b) \, \cU_{t}(b+x_{0}) \, \Big( 1 +  \e^{- \int_0^t \d s \, V(b(s))} \Big) \right] ,
\end{equation}
where $\Xi_{x}^{t}$ denotes the event that $\sup_{s \in [0,t]} |b(s) -x| > \dist\big(x, \supp(V)\big)$.
The Cauchy-Schwarz inequality and \eqref{BHL} then imply
\begin{equation}\label{eq:upperbd2}
 G(x) \le c \left(\PP_{x,x}^{0,t}[\Xi_{x}^{t}] \right)^{1/2}=  c \left(\PP_{0,0}^{0,t} \bigg[ \max_{s \in [0,t]} |b(s)| > \dist\big(x, \supp(V)\big) \bigg] \right)^{1/2}
\end{equation}
for an $x$-independent constant $c \in ]0,\infty[$. The probability
in the last line is the complement of the distribution function for
the maximum of a $d$-dimensional Bessel bridge \cite{YoZa04} and has
a Gaussian decay, see, for example, \cite[p.~341, Lemma 1]{kirsch88}
or \cite[p.~438]{GrSh96}. This yields $G \in \mathrm{L}^{1}(\Rd)$.

\par\noindent
6. We now combine (\ref{inner-limit}) and the upper bound (\ref{eq:upperbd1}), (\ref{eq:upperbd2})
to evaluate the limit of \eqref{conv-start}. Dominated convergence applied to \eqref{conv-start} gives for
all $(x_{L_{n}})_{n\in\NN} \in \AA(x_{0})$
\begin{equation}
\label{limit}
   \lim_{n\to\infty} \varphi^{(n)}_{x_{L_{n}}}(t)  = \int_{\Rd}\!\d x\, \lim_{n\to\infty}
   F_{n}(x) = \int_{\Rd}\!\d x\, \EE_{x,x}^{0,t} \big[ \cU_{t}(b) \, \cV_{t}(b-x_{0}) \big],
\end{equation}
where we have used another change of variables. Therefore Theorem~\ref{uniform} follows from
\begin{equation}
\begin{split}
\lim_{n\to\infty}   \wtilde{\xi}^{(n)}_{x_{L_{n}}}(t) &=  \frac{1}{t} \; (2\pi t)^{-d/2}  \int_{\Rd}\!\d x\, \EE_{x,x}^{0,t} \big[ \cU_{t}(b) \, \cV_{t}(b-x_{0}) \big] \\
     &=  \frac{1}{t} \;\tr_{\mathrm{L}^{2}(\Rd)} \left[ \exp(-t H_{0}) - \exp(-t H_{1, x_{0}}) \right] \\
 &= \int_{\RR}\!\d E\, \e^{-tE} \,
 \xi\big(E; H_{1,x_{0}}, H_{0}\big)
\end{split}
\end{equation}
for all $(x_{L_{n}})_{n\in\NN} \in \AA(x_{0})$ and from Theorem~\ref{cont-laplace}.
\end{proof}


\begin{acknowledgement}
  PM would like to thank Nariyuki Minami, Shu Nakamura, Fumihiko Nakano and
  Naomasa Ueki for their splendid hospitality during the \emph{Kochi School on Random
  Schr\"odinger operators} and the RIMS-workshop \emph{Spectra of Random Operators and Related Topics}
  in Kyoto in November and December 2009. PM also acknowledges Edgardo Stockmeyer and Ingo Wagner
  for interesting discussions. The authors thank the Centre
  Interfacultaire Bernoulli in Lausanne, Switzerland, for hospitality and the Swiss NSF for
  partial financial support during the time this work was finished.
\end{acknowledgement}



\begin{thebibliography}{99}
\frenchspacing

\bibitem{AiSi82}
  \au{M.}{Aizenman}\et\lau{B.}{Simon}
  \ti{Brownian motion and Harnack inequality for Schr\"odinger operators}
  \z{Commun. Pure Appl. Math.}{35}{209--273}{1982}

\bibitem{BiYa93}
  \au{M. Sh.}{Birman}\et\lau{D. R.}{Yafaev}
  \ti{The spectral shift function. The papers of M. G. Kre\u{\i}n
    and their further development}
  \z{St. Petersburg Math. J.}{4}{833--870}{1993}

\bibitem{BoMa09}
    \au{V.}{Borovyk}\et\lau{K. A.}{Makarov}
    \ti{On the weak and ergodic limit of the spectral shift function}
    preprint    arXiv:0911.3880.

\bibitem{BrHuLe00}
  \au{K.}{Broderix}, \au{D.}{Hundertmark}\et\lau{H.}{Leschke}
  \ti{Continuity properties of Schr\"odinger semigroups with magnetic fields}
  \z{Rev. Math. Phys.}{12}{181--225}{2000}

\bibitem{Fel71}
  \au{W.}{Feller}
  \bti[vol. 2, 2nd ed.]{An introduction to probability theory and its
    applications}
  \pub{Wiley}{New York}{1971}

\bibitem{GeKoSc95}
  \au{R.}{Geisler}, \au{V.}{Kostrykin}\et\lau{R.}{Schrader}
  \ti{Concavity properties of Krein's spectral shift function}
  \z{Rev. Math. Phys.}{7}{161--181}{1995}

\bibitem{GrSh96}
  \au{J.-C.}{Gruet}\et\lau{Z.}{Shi}
  \ti{The occupation time of Brownian motion in a ball}
  \z{J. Theor. Probab.}{9}{429--446}{1996}

\bibitem{HiKlMu10}
  \au{P. D.}{Hislop}, \au{A.}{Klein}\et\lau{P.}{M\"uller}
  \ti{A lower bound for the density of states of the continuum Anderson model in the localization regime}
  in preparation.

\bibitem{HiMu08}
  \au{P. D.}{Hislop}\et\lau{P.}{M\"uller}
  \ti{A lower bound for the density of states of the lattice Anderson model}
  \z{Proc. Amer. Math. Soc.}{136}{2887--2893}{2008}

\bibitem{HiMu10}
  \au{P. D.}{Hislop}\et\lau{P.}{M\"uller}
  \ti{The spectral shift function for
  compactly supported perturbations of Schr\"odinger operators on large bounded domains}
  \z{Proc. Amer. Math. Soc.}{138}{2141--2150}{2010}

\bibitem{Kir87}
  \lau{W.}{Kirsch}
  \ti{Small perturbations and the
    eigenvalues of the Laplacian on large bounded domains}
  \z{Proc. Amer. Math. Soc.}{101}{509--512}{1987}

\bibitem{Kir89}
  \lau{W.}{Kirsch}
  \ti{The stability of the density of states of Schr\"odinger operators under
    very small perturbations}
  \z{Int. Eqns. Op. Th.}{12}{383--391}{1989}

\bibitem{kirsch88}
\lau{W.}{Kirsch}
  \ti{Random Schr\"odinger operators. A course}
  \bti{Schr\"odinger operators\textup{, S\o{}nderborg 1988}}
  \z{Lecture Notes in Physics\textup{, Springer, Berlin}}{345}{264--370}{1989}

\bibitem{Sim79}
  \lau{B.}{Simon}
  \bti{Functional integration and quantum physics}
  \pub{Academic}{New York}{1979}

\bibitem{Sim82}
  \lau{B.}{Simon}
  \ti{Schr{\"o}dinger semigroups}
  \z{Bull. Amer. Math. Soc. (N.S.)}{7}{447--526}{1982}
  Erratum: \z{ibid.}{11}{426}{1984}

\bibitem{Yaf92}
  \lau{D. R.}{Yafaev}
  \bti{Mathematical scattering theory. General theory}
  \pub{Amer. Math. Soc.}{Providence, RI}{1992}

\bibitem{YoZa04}
  \au{M.}{Yor}\et\lau{L.}{Zambotti}
  \ti{A remark about the norm of a Brownian bridge}
  \z{Stat. Probab. Lett.}{68}{297--304}{2004}

\end{thebibliography}
\end{document}